\documentclass[11pt]{article}
\usepackage{hyperref}
\usepackage{fullpage}
\usepackage{latexsym}
\usepackage{theorem}
\usepackage{amstext}
\usepackage{xspace}
\usepackage{verbatim}
\usepackage{amssymb}
\usepackage{amsmath}
\usepackage{ifthen}

\newtheorem{theorem}{Theorem}[section]

\newtheorem{remark}{Remark}
\newtheorem{proposition}{Proposition}

\newenvironment{proof}{\noindent {\bf Proof.  }}{\hfill$\Box$}

\def \poly { \text{\rm poly} }

\def \C {{\mathbb C}}
\def \Z {{\mathbb Z}}

\def \Sig[#1] {{\sf \Sigma}_{#1} }

\def \Pie[#1] {{\sf \Pi}_{#1} }

\begin{document}

\title{Finding a path of length $k$ in $O^*(2^k)$ time}

\author{Ryan Williams\thanks{Current address: School of Mathematics, Institute for Advanced Study, Princeton, NJ. Email: {\tt ryanw@math.ias.edu}. This research was partially supported by the National Science Foundation under CCR-0122581, while the author was at Carnegie Mellon.}\\Carnegie Mellon University}

\date{}

\maketitle

\begin{abstract} We give a randomized algorithm that determines if a given graph has a simple path of length at least $k$ in $O(2^k \cdot \poly(n,k))$ time.\end{abstract}


\section{Introduction}

The $k$-path problem is to determine if a given graph contains a simple path of length at least $k$, and if so, produce such a path. When $k$ is given as a part of the input, the problem is well-known to be {\sf NP}-complete. The general problem has many practical applications (cf.~\cite{TSP,Scott}).  The trivial algorithm enumerating all possible $k$-paths in an $n$-node graph uses $\Theta(n^k)$ time, so it is only polynomial for $k = O(1)$. The first algorithm to reduce the runtime dependency on $k$ was given by Monien~\cite{Monien} whose algorithm runs in $O^*(k!)$ time (the $O^*$ notation suppresses $\poly(n,k)$ factors). Hence the case $k \leq (\log n)/(\log \log n)$ is still polynomial time solvable. For years it was not known if the $O(\log n)$-path problem was in polynomial time; a breakthrough by Alon, Yuster, and Zwick~\cite{AYZ} finally led to such an algorithm. They gave a randomized algorithm running in $O^*((2e)^k) \leq O^*(5.44^k)$ time, and a deterministic $O^*(c^k)$ time algorithm, where $c$ is a large constant. Since it has been known for many years prior that when $k=n$ the problem is solvable in $O^*(2^k)$ time~\cite{Bellman,Held,Karp}, it is natural to ask if there is an algorithm that can match this runtime for all values of $k$.

It has been only recently that faster $k$-path algorithms have appeared in the literature. In 2006, two groups independently discovered $O^*(4^k)$ randomized algorithms and $O^*(c^k)$ deterministic algorithms, with $c = 16$~\cite{Kneis} and $c = 12.5$~\cite{Chen}. Very recently, Koutis~\cite{Koutis} presented a novel randomized algorithm for $k$-path that runs in $O^*(2^{3k/2}) \leq O^*(2.83^k)$ time. In this note, we extend his result to obtain an $O^*(2^k)$ time algorithm. Koutis shows how to detect if a graph has a $k$-subgraph with an odd number of $k$-paths in $O^*(2^k)$ time. By augmenting his approach with more random choices and some additional ideas, we can find a $k$-path in roughly the same running time. As mentioned above, the best known algorithms for finding a Hamilton path in an $n$-node graph run in $O^*(2^n)$ time and are quite old. Therefore any significant improvement in the runtime dependence on $k$ given by our algorithm would imply a faster Hamilton path algorithm, which has been an open problem for over forty years. We do not wish to insist that our algorithm is optimal, but rather that further progress would entail a substantial breakthrough in algorithms for {\sf NP}-hard problems.

\section{Preliminaries}

Most of our notation is standard, however we do require some notions from algebra that are not often used in graph algorithms. Let $F$ be a field and $G$ be a multiplicative group (its binary operation is written as a multiplication). The group algebra $F[G]$ is an algebraic object that incorporates properties of both objects, defined as follows. Elements of $F[G]$ have the form \[\sum_{g \in G} a_g g,\] where each $a_g \in F$. That is, the elements are formal sums over the group elements, with coefficients from the field. Addition in $F[G]$ is defined in a point-wise manner: \[\left(\sum_{g \in G} a_g g\right) + \left(\sum_{g \in G} b_g g\right) = \sum_{g \in G}(a_g + b_g)g.\] Multiplication has the form of a convolution: \[\left(\sum_{g \in G} a_g g\right)\cdot \left(\sum_{g \in G} b_g g\right) = \sum_{g \in G}\left(\sum_{h \in G} a_{h}b_{h^{-1}g}\right)g.\] Note the above definition coincides with the one in~\cite{Koutis}. The above operations define a ring with $0$ and $1$, where $0 \in F[G]$ is the element $\sum_{g \in G} a_g g$ such that all $a_g$ are equal to $0 \in F$, and $1 \in F[G]$ is the multiplicative identity $1 \in G$ of the group.

In our algorithm, we work over the group algebra $GF(2^{\ell})[\Z_2^{k}]$, for particular integers $k,\ell \geq 0$. Here $\Z_2^{k}$ is the group of binary $k$-vectors, endowed with componentwise addition modulo $2$ as its operation.  $GF(2^{\ell})$ is the unique field on $2^{\ell}$ elements. We use $W_0$ to denote the all-zeros vector (the identity) of $\Z_2^k$. Note that every $v \in \Z_2^{k}$ is its own inverse: $v^2 = W_0$. Every element in the algebra has the form $\sum_{v \in \Z_2^{k}} a_v v$, where $a_v \in GF(2^{\ell})$.

\paragraph{Example} The elements of $GF(2^2)$ can be represented as the four polynomials $0,1,x,1+x$ over $GF(2)$, where computations are done modulo $x^2+x+1$. For example, $x^3 = x\cdot x^2 = x \cdot (1+x) = x + x^2 = 1$ in $GF(2^2)$. Over $GF(2^2)[\Z^{3}_2]$, \[\left(\left[\begin{array}{c} 0 \\ 0 \\ 0 \end{array}\right] + x \left[\begin{array}{c} 1\\ 0 \\ 1 \end{array}\right]\right) + \left(\left[\begin{array}{c} 0 \\ 0 \\ 0 \end{array}\right] + \left[\begin{array}{c} 1 \\ 0 \\ 1 \end{array}\right] + \left[\begin{array}{c} 1 \\ 1 \\ 1 \end{array}\right] \right) = (1+x)\left[\begin{array}{c} 1 \\ 0 \\ 1 \end{array}\right] + \left[\begin{array}{c} 1 \\ 1 \\ 1 \end{array}\right]\] and over $F[\Z^3_2]$ in general, \[\left(a_1 \left[\begin{array}{c} 0 \\ 0 \\ 0 \end{array}\right] + a_2 \left[\begin{array}{c} 1\\ 0 \\ 1 \end{array}\right]\right) \cdot \left(b_1 \left[\begin{array}{c} 0 \\ 0 \\ 0 \end{array}\right] + b_2 \left[\begin{array}{c} 1 \\ 0 \\ 1 \end{array}\right] + b_3 \left[\begin{array}{c} 1 \\ 1 \\ 1 \end{array}\right] \right)\]\[~~~ = (a_1 b_1 + a_2 b_2)\left[\begin{array}{c} 0 \\ 0 \\ 0 \end{array}\right] + a_2 b_3 \left[\begin{array}{c} 0 \\ 1 \\ 0 \end{array}\right] + (a_1 b_2 + a_2 b_1)\left[\begin{array}{c} 1 \\ 0 \\ 1 \end{array}\right] + a_1 b_3 \left[\begin{array}{c} 1 \\ 1 \\ 1 \end{array}\right].\]

\section{Algorithm for the $k$-Path Problem}

Fix an underlying graph $G$ in the following, with vertex set $\{1,\ldots,n\}$. Let $F$ be a field, let $A$ be the adjacency matrix of $G$, and let $x_1,\ldots,x_n$ be variables. Define a matrix $B[i,j] = A[i,j] x_i$. Let $\vec{\bf 1}$ be the row $n$-vector of all $1$'s, and $\vec{x}$ be the column vector defined by $\vec{x}[i] = x_i$. Define the $k$-walk polynomial to be $P_k(x_1,\ldots,x_k) = \vec{\bf 1} \cdot B^{k-1} \cdot \vec{x}$.

\begin{proposition} \[P_k(x_1,\ldots,x_k) = \sum_{i_1,\ldots,i_k\text{ is a walk in }  G} x_{i_1} \cdots x_{i_k}.\]
\end{proposition}

Clearly, there is a $k$-path in $G$ iff $P_k(x_1,\ldots,x_n)$ contains a multilinear term. We give a randomized algorithm $R$ with the property that:

\begin{itemize}
\item if $P_k$ has a multilinear term, then $\Pr[R~\text{outputs {\em yes}}] \geq 1/5$,

\item if $P_k$ does not have a multilinear term, then $R$ always outputs {\em no}.
\end{itemize}

In fact, the statement we can prove is more general.

\begin{theorem} Let $P(x_1,\ldots,x_n)$ be a polynomial of degree at most $k$, represented by an arithmetic circuit of size $s(n)$ with $+$ gates (of unbounded fan-in), $\times$ gates (of fan-in two), and no scalar multiplications. There is a randomized algorithm that on every $P$ runs in $O^*(2^k s(n))$ time, outputs {\em yes} with high probability if there is a multilinear term in the sum-product expansion of $P$, and always outputs {\em no} if there is no multilinear term.
\end{theorem}

\begin{remark}\label{rem} We may assume without loss of generality that every multilinear term of $P$ has degree at least $k$, and at least one multilinear term has degree exactly $k$. If not, let $k' < k$ be the minimum degree of a multilinear term in $P$. We can try all $j = 1,\ldots,k$ and multiply the final output of the circuit for $P$ by $j$ new variables $x_{n+1},\ldots,x_{n+j}$, obtaining a polynomial $P^j$, which we feed to the randomized algorithm. Note that when $j = k-k'$, our assumption holds.
\end{remark}

By observing that $P_k$ can be implemented with a circuit of size $O(k(m+n))$ where $m$ is the number of edges in $G$, the $k$-path algorithm is obtained. We begin the proof with a description of the algorithm. The basic idea is to substitute random group algebra elements for the variables such that all non-multilinear terms in $P$ evaluate to zero and some multilinear terms survive. Then we augment the scalar-free multiplication circuit with random scalar multiplications over a field large enough that the remaining multilinear polynomial evaluates to nonzero with decent probability. Set $F = GF(2^{3 + \log k})$.

\paragraph{Algorithm} {\em Pick $n$ uniform random vectors $v_1,\ldots,v_n$ from $\Z_2^k$. For each multiplication gate $g_i$ in the circuit for $P$, pick a uniform random $w_i \in F\setminus\{0\}$. Insert a new gate that multiplies the output of $g_i$ with $w_i$, and feeds the output to those gates that read the output of $g_i$. Let $P'$ be the new polynomial represented by this arithmetic circuit.\footnote{In the evaluation of the $k$-path polynomial $P_k$, the algorithm corresponds to picking random $y_{i,j,c}$ in $F$ for $c=1,\ldots,k-1$, $i,j=1,\ldots,n$, letting $B_c[i,j] = y_{i,j,c}B[i,j]$, then evaluating $P'_k(x_1,\ldots,x_n) = \vec{\bf 1} \cdot B_{k-1} \cdots B_1 \cdot \vec{x}$ on the appropriate vectors.} Output {\em yes} iff $P'(W_0+v_1,\ldots,W_0+v_n) \neq 0$.}

\paragraph{Runtime} Let us describe one way to implement the algorithm efficiently. The only non-trivial step is the final polynomial evaluation. By definition, the evaluation of $P'(W_0+v_1,\ldots,W_0+v_n)$ takes $O(s(n))$ arithmetic operations. However, since evaluation takes place over $F[\Z_2^{k}]$, we need to account for the cost of arithmetic in the group algebra. Elements in $F[\Z_2^{k}]$ can be naturally interpreted as vectors in $F^{2^k}$. Addition of these vectors (as elements in $F[\Z_2^{k}]$) can be done in $O(2^k \log |F|)$ time with a component-wise sum. Multiplication of vectors $u$ and $v$ over the group algebra can be done in $O(k 2^k \log^2 |F|)$ time by a Fast Fourier Transform style algorithm, as we now describe.

For simplicity, let $\ell = 3 + \log_2 k$. Represent elements of $F = GF(2^{\ell})$ as univariate polynomials over $GF(2)$ of degree at most $\ell$ in the usual way, so the entries of $u$ and $v$ are degree-$\ell$ polynomials. Over the ring $\C[x]$, multiply $u$ and $v$ with the matrix $H_{k}$ for the discrete Fourier transform on $\Z_2^{k}$ (also called the Walsh-Hadamard transform) in $O(k 2^k M(\ell))$ time (cf. \cite{MaslenRockmore}) where $M(d)$ is the runtime for computing the product of two degree-$d$ univariate polynomials over $GF(2)$. Since $\ell$ is small, it suffices to use the bound $M(\ell) \leq O(\ell^2)$. Take the pointwise product of the two resulting vectors obtaining a vector $w$, and multiply $w$ with $H_{k}$ (note $H_{k}^{-1} = H_{k}$, so this is the inverse of the transform). The resulting vector $x$ contains $2^k$ polynomials of degree at most $2\ell$. Reduce each polynomial modulo an irreducible degree-$\ell$ polynomial over $GF(2)$, in $O(2^k M(\ell))$ time. (For a discussion of how to obtain irreducible polynomials, cf.~\cite{Shoup}.) This has the effect of mapping our results in $\C[x]$ back down to $GF(2^{\ell})$. Overall, the evaluation of $P'$ takes at most $O^*(2^k s(n))$ time.

We note that while the above computation naively needs $\Omega(2^k)$ space, detecting if $P'$ evaluates to zero can be done in $O(\poly(n,k))$ space using (for example) the representation-theoretic technique of Koutis~\cite{Koutis}. For the sake of brevity, we will not concern ourselves with this issue.

\paragraph{Correctness} The crucial observation of Koutis~\cite{Koutis} is that, for any $v_i \in \Z_2^k$,  \[(W_0+v_i)^2 = W_0^2 + 2 v_i + v_i^2 = W_0 + 0 + W_0 = 0 \mod 2.\] Therefore all squares in $P$ vanish in $P'(W_0+v_1,\ldots,W_0+v_n)$, since $F$ has characteristic $2$. It follows that if $P(x_1,\ldots,x_n)$ does not have a multilinear term, then $P'(W_0+v_1,\ldots,W_0+v_n) = 0$ over $F[\Z_2^k]$, regardless of the choice of $v_i$.

In the remaining paragraphs, we prove that if the sum-product expansion of $P(x_1,\ldots,x_n)$ has a multilinear term, then $P'(W_0+v_1,\ldots,W_0+v_n) \neq 0$ with probability at least $1/5$, over the random choices of $w_i$'s and $v_i$'s. By Remark~\ref{rem}, we may assume that every multilinear term in the sum-product expansion of $P$ has the form $c \cdot x_{i_1} \cdots x_{i_{k'}}$ where $k' \geq k$ and $c \in \Z$. For each such term, there is a corresponding collection of multilinear terms in $P'$, each of the form \[w_1\cdots w_{k'-1} \prod_{j=1}^{k'} (W_0+v_{i_j}),\] where the sequence $w_1,\ldots, w_{k'-1}$ is distinct for every term in the collection (as the sequences of multiplication gates $g_1,\ldots,g_{k'-1}$ are distinct). Note these terms do not have leading coefficients, since there are no scalar multiplications in the arithmetic circuit.

Consider a monomial $\prod_{j=1}^i (W_0+v_{j})$ in the polynomial (disregarding the $w_i$'s for the moment). Koutis~\cite{Koutis} proved that if the $i$ vectors $v_1,\ldots,v_i$ are linearly dependent, this monomial vanishes modulo $2$. We observe that his proof works over any field of characteristic two.

\begin{proposition}[Koutis] \label{dep} If $v_{1},\ldots,v_{i} \in \Z_2^{k}$ are linearly dependent over $GF(2)$, then $\prod_{j=1}^i (W_0+v_{j}) = 0$ in $F[\Z_2^{k}]$.\end{proposition}

\begin{proof} If $v_1,\ldots,v_i$ are linearly dependent, there is a nonempty subset $T$ of the vectors that sum to the all-zeros vector. In $F[\Z_2^{k}]$, this is equivalent to \[\prod_{j \in T} v_j = W_0.\] Let $S \subseteq T$ be arbitrary. Multiplying both sides by $\prod_{j \in (S \Delta T)} v_j$, \[\prod_{j \in S} v_j = \prod_{j \in (S \Delta T)} v_j.\] Therefore $\prod_{j \in T} (W_0+v_{j}) = \sum_{S \subseteq T} \left(\prod_{j \in S} v_{j}\right) = 0 \mod 2$, since each product appears twice in the sum. Hence $\prod_{j=1}^i (W_0+v_{j}) = 0$ over $F[\Z_2^k]$, since $F$ is characteristic $2$.\end{proof}

Therefore linearly dependent vectors lead to a cancellation of terms. On the other hand, when $v_1,\ldots,v_i$ are linearly independent, $\prod_{j=1}^i (W_0+v_{j})$ is just the sum over all vectors in the span of $v_1,\ldots,v_i$, since each vector in the span is of the form $\prod_{j \in S} v_j$ for some $S \subseteq [i]$, and there is a unique way to generate each vector in the span.

\begin{proposition} \label{indep} If $v_1,\ldots,v_k \in \Z_2^{k}$ are linearly independent over $GF(2)$, then $\prod_{j=1}^k (W_0+v_j) = \sum_{v \in \Z_2^k} v$.\end{proposition}

By Propositions~\ref{dep} and \ref{indep}, and the fact that any $k' > k$ vectors are linearly dependent, $P'(W_0+v_1,\ldots,W_0+v_n)$ evaluates to either $0$, or $c \sum_{v \in \Z_2^{k}} v$ for some $c \in F$. The final piece of our argument shows that if $P$ has a multilinear term, then $c \neq 0$ with probability at least $1/5$.

The vectors $v_{\ell_1},\ldots,v_{\ell_k}$ chosen for the variables in a multilinear term of $P$ are linearly independent with probability at least $1/4$, because the probability that a random $k \times k$ matrix over $GF(2)$ has full rank is at least $0.28 \geq 1/4$~(cf.~\cite{BlumKannan},~Lemma~6.3.1). Hence in $P'(W_0+v_1,\ldots,W_0+v_n)$, there is at least one multilinear term in $P$ corresponding to a set of $k$ linearly independent vectors, with probability at least $1/4$.

Let $S$ be the set of those multilinear terms in $P$ which correspond to $k$ linearly independent vectors in $P'(W_0+v_1,\ldots,W_0+v_n)$. Then the above coefficient $c = \sum_{i} c_i$ for some $c_i \in F$ corresponding to the $i$th multilinear term in $S$. Conditioned on $S \neq \varnothing$, we claim that $\sum_{i} c_i = 0$ with probability at most $1/2^3$. Each coefficient $c_i$ comes from a sum of products of $k-1$ elements $w_{i,1},\ldots,w_{i,k-1}$ corresponding to some multiplication gates $g_{i,1},\ldots,g_{i,k-1}$ in the circuit. (In the $k$-path case, note that each $c_i$ is a sum of products of the form $y_{i_1,i_2,1} y_{i_2,i_3,2} \cdots y_{i_{k-1},i_k,k-1}$.) Construing the $w_i$'s as variables, the sum $Q(w_1,\ldots,w_{s(n)}) = \sum_i c_i$ is a {\em degree-$k$ polynomial over $F$} in the variables. Assuming $S \neq \varnothing$, $Q$ is not identically zero. (Note each monomial in $Q$ has coefficient $1$.) By the Schwartz-Zippel Lemma~\cite{MR}, the algorithm's random assignment to the variables of $Q$ results in an evaluation of $0 \in F$ with probability at most $k/|F| = 1/2^3$. Since $S \neq \varnothing$ with probability at least $1/4$, the overall probability of success is at least $1/4 \cdot (1 - 1/2^3) > 1/5$.

\paragraph{Constructing a Path} The algorithm $R$ merely detects if a graph has a $k$-path. We note that an $O^*(2^k)$ algorithm producing a $k$-path (when it exists) is easy to obtain; let us briefly outline one possible algorithm of this kind. For an arbitrary node $v_i$, we remove $v_i$ from the graph and run the $k$-path detection algorithm for $O(\log n)$ trials, using new random bits for each trial. If the algorithm outputs yes in some trial, we recursively call our algorithm on the graph with $v_i$ removed, returning the $k$-path that it returns. Otherwise, we add $v_i$ back to the graph and move to the next candidate node $v_{i+1}$, noting that such a move occurs at most $k$ times (with high probability). Hence we can bound the runtime with the recurrence \[T(n) \leq O^*(2^k \cdot k \log n) + T(n-1),\] which is $O^*(2^k)$. The overall probability of error can be bounded by a constant less than $1$, since the probability that all $O(\log n)$ trials result in error is inversely polynomial in $n$.

\section{Conclusion}

We end with two interesting open questions. We conjecture that both can be answered affirmatively.
\begin{itemize}

\item Let $G$ be a graph with arbitrary costs on its edges. The {\sc Short Cheap Tour} problem is to find a path of length at least $k$ where the total sum of costs on the edges is minimized. This problem is fixed-parameter tractable, in fact: \begin{theorem} {\sc Short Cheap Tour} can be solved in $O^*(4^k)$ time by a randomized algorithm that succeeds with high probability. \end{theorem} We omit the proof here; our algorithm is a simple extension of the divide-and-color approach for solving $k$-path~\cite{Kneis}. Can {\sc Short Cheap Tour} be solved in $O^*(2^k)$ time? The algorithm of this paper does not appear to extend to weighted graphs.

\item Is there a {\em deterministic} algorithm for $k$-path with the same runtime complexity as our algorithm? This question was also raised by Koutis~\cite{Koutis}, however our algorithm looks more difficult to derandomize. Our argument relies on the fact that polynomial identity testing is in {\sf RP}, and it is known that a polytime derandomization of this would imply strong circuit lower bounds~\cite{IK}.

\end{itemize}

\section{Acknowledgements}

I am very grateful to Yiannis Koutis for sharing an early preprint of his paper, and for several valuable discussions on his work. I also thank Andreas Bj\"{o}rklund and the anonymous referees for useful comments.

\end{document}